\documentclass{article}

\usepackage[a4paper]{geometry}
\usepackage{amsmath}
\usepackage{amssymb}
\usepackage{amsthm}
\usepackage[ruled, linesnumbered, vlined]{algorithm2e}
\usepackage{mathtools}
\usepackage{latexsym}
\usepackage{biblatex}
\usepackage{tikz}

\usetikzlibrary{positioning,decorations.pathreplacing}

\newtheorem{theorem}{Theorem}[section]
\newtheorem{corollary}[theorem]{Corollary}
\newtheorem{lemma}[theorem]{Lemma}

\newcommand{\ket}[1]{|#1\rangle}
\newcommand{\LBDofFraction}{\lambda}

\title{Quantum Approximate Counting for Markov Chains and Application to Collision Counting}

\author{
Fran{\c c}ois Le Gall\footnote{legall@math.nagoya-u.ac.jp}\\
Graduate School of Mathematics\\ 
Nagoya University\\
\and
Iu-Iong Ng\footnote{m20048d@math.nagoya-u.ac.jp}\\
Graduate School of Mathematics\\ 
Nagoya University\\
}

\date{}

\begin{document}

\maketitle

\begin{abstract}
In this paper we show how to generalize the quantum approximate counting technique developed by Brassard, H{\o}yer and Tapp [ICALP 1998] to a more general setting: estimating the number of marked states of a Markov chain (a Markov chain can be seen as a random walk over a graph with weighted edges). This makes it possible to construct quantum approximate counting algorithms from quantum search algorithms based on the powerful ``quantum walk based search'' framework established by Magniez, Nayak, Roland and Santha [SIAM Journal on Computing 2011]. As an application, we apply this approach to the quantum element distinctness algorithm by Ambainis [SIAM Journal on Computing 2007]: for two injective functions over a set of $N$ elements, we obtain a quantum algorithm that estimates the number $m$ of collisions of the two functions within relative error $\epsilon$ by making $\tilde{O}\left(\frac{1}{\epsilon^{25/24}}\big(\frac{N}{\sqrt{m}}\big)^{2/3}\right)$ queries, which gives an improvement over the $\Theta\big(\frac{1}{\epsilon}\frac{N}{\sqrt{m}}\big)$-query classical algorithm based on random sampling when $m\ll N$. 
\end{abstract}

\section{Introduction}
\subparagraph{Quantum search.}
Grover's search \cite{GroverSTOC96} is the most basic version of quantum search. This quantum algorithm considers the following problem called \emph{unstructured search}: for a function $f\colon\{1,\ldots,N\}\to\{0,1\}$ given as a black-box, find one element $x\in\{1,\ldots,N\}$ such that $f(x)=1$ (we call such an element a solution and assume that there exists at least one). Let $M=|f^{-1}(1)|$ denote the number of solutions. Grover's algorithm outputs with high probability a solution using $O(\sqrt{N/M})$ calls to the function~$f$. This is a quadratic improvement over the classical (i.e., non-quantum) strategy, which is based on random sampling and requires $\Theta(N/M)$ calls to $f$ to find a solution with high probability.

Another fundamental, but more sophisticated, quantum search technique is \emph{quantum walk based search}~\cite{Ambainis07,MagniezNRS11,Szegedy+FOCS04}. The problem considered here is a generalization of unstructured search. As for unstructured search, this problem accesses the data via a black-box. It can be defined in an abstract way for any Markov chain (Markov chains are generalizations of the concept of random walks over a graphs --- see Section~\ref{sec:Markov} for details). Some states of the Markov chains corresponding to ``solutions'' are called the marked states. The goal of the problem is to find one marked state (here we assume again that there is at least one marked state). Let $\LBDofFraction>0$ be a lower bound on the fraction of marked states.  Magniez, Nayak, Roland and Santha~\cite{MagniezNRS11} have shown that this problem can be solved with high probability by a quantum algorithm that makes 
\begin{equation}\label{eq1}
O\left(\mathsf{S}+\frac{1}{\sqrt{\LBDofFraction}}\left(\frac{1}{\sqrt{\delta}}\mathsf{U}+\mathsf{C}\right)\right)
\end{equation}
queries to the black-box, were $\delta$ denotes the spectral gap of the Markov chain and $\mathsf{S}$, $\mathsf{U}$ and~$\mathsf{C}$ represent the number of queries needed to implement the three basic quantum operations corresponding to search: the setup operation, the update operation and the checking operation. Unstructured search simply corresponds to the Markov chain representing a random walk on a complete graph of $N$ vertices, which has spectral gap $\delta=1$, with $\LBDofFraction=M/N$, $\mathsf{S}=1$, $\mathsf{U}=2$ and $\mathsf{C}=0$. The complexity becomes $O(\sqrt{\LBDofFraction})=O(\sqrt{N/M})$, as for Grover's search.\footnote{Grover's search can also be derived in another (more direct) way, with costs $\mathsf{S}=0$, $\mathsf{U}=0$ and $\mathsf{C}=1$.}
\, A striking illustration of the power of quantum walk based search is the problem called \emph{element distinctness}: for two functions $f, g\colon\{1,\ldots,N\}\longrightarrow \{1,\ldots, K\}$ accessible as black boxes, find a pair $(i, j)\in \{1,\ldots,N\}^2$ satisfying $f(i)=g(j)$ if such a pair exists (such a pair is called a collision). Ambainis~\cite{Ambainis07} showed how to solve this problem with $O(N^{2/3})$ queries using quantum walk based search.\footnote{While the version of element distinctness considered in~\cite{Ambainis07} uses only one function, the version we present here is essentially equivalent (see \cite{LeGall+20} for details).}
\, Besides element distinctness, quantum walk based search has been used to design quantum algorithms for many other problems, such as matrix multiplication~\cite{Buhrman+SODA06}, triangle finding~\cite{Jeffery+SODA13,LeGallFOCS14, Magniez+SICOMP07} or associativity testing~\cite{Lee+17}. 

\subparagraph{Quantum approximate counting.}
\emph{Quantum approximate counting} is another very useful quantum primitive designed by Brassard, H{\o}yer and Tapp~\cite{Brassard+98}. It solves the (approximate) counting version of unstructured search: for any given $\epsilon>0$ and a function $f\colon\{1,\ldots,N\}\to\{0,1\}$ given as a black-box, output an $\epsilon$-approximation of the number of solution $M=|f^{-1}(1)|$, i.e., output an integer $\hat M$ such that $|M-\hat M|\le \epsilon M$. The quantum approximate counting algorithm from~\cite{Brassard+98} solves this problem using $O(\frac{1}{\epsilon}\sqrt{N/M})$ calls to~$f$. This is again a quadratic improvement over the classical strategy, which is based on random sampling and requires $\Theta(\frac{1}{\epsilon^2}N/M)$ calls to $f$~\cite{Nayak+99}. Very recently, several variants of quantum approximate counting have been proposed~\cite{Aaronson+20,Suzuki+20,Venkateswaran+STACS21,Wie19}. These variants have the same asymptotic complexity as the original algorithm but may be easier to implement in practice, especially on near-future quantum computers.

Since quantum walk based search is a generalization of Grover search, a natural question is whether a version of quantum approximate counting can be obtained for quantum walk based search as well, i.e., whether there exists an efficient quantum algorithm to estimate the number of marked states of a Markov chain. To our knowledge, this problem has never been studied in the literature. This problem is especially motivated by the task of approximate collision counting: for two functions $f, g\colon\{1,\ldots,N\}\longrightarrow \{1,\ldots, K\}$ accessible as black boxes, output an approximation of the number of collisions. While not stated explicitly, this task is at the heart of most algorithms for estimating the edit distance between two non-repetitive strings~\cite{Andoni+SODA10,LeGall+20,Naumovitz+SODA17}, and is also used for estimating the $\ell_1$-distance between two probability distributions~\cite{Batu+13,Valiant11}.

\subparagraph{Description of our results.}

Our main technical result is the following general approximate counting version for reversible Markov chains.
\vspace*{12pt}
\begin{theorem}
\label{thm:Count}
Let $P$ be a reversible ergodic Markov chain with uniform stationary distribution, $\delta>0$ be the spectral gap of $P$ and $\LBDofFraction>0$ be a known lower bound on the fraction of marked states of $P$. Let $\mathsf{S}$, $\mathsf{U}$ and $\mathsf{C}$ be the number of queries needed to implement, respectively,  the setup operation, the update operation and the checking operation. For any $\epsilon\in (0, 1)$, there exists a quantum algorithm which outputs with high probability an estimate $\hat{M}$ such that $|M-\hat{M}|<\epsilon M$, and uses 
\[
\tilde O\left(\mathsf{S}+\frac{1}{\epsilon}\frac{1}{\sqrt{\LBDofFraction}}\left(\frac{1}{\sqrt{\delta}}\mathsf{U}+\mathsf{C}\right)\right)
\] 
queries to the black box.\footnote{In this paper the notation $\tilde O(\cdot)$ removes the polylogarithmic factors in all parameters.}
\end{theorem}
\vspace*{12pt}

As an application of Theorem~\ref{thm:Count}, we show how to construct a quantum algorithm to approximate the number of collisions. More precisely, the version of approximate collision counting we consider is as follows: for $\epsilon>0$ and for two injective functions $f, g\colon\{1,\ldots,N\}\longrightarrow \{1,\ldots, K\}$, where $K>N$, output an $\epsilon$-approximation of the number of collisions $m$, i.e., an estimate $\hat{m}$ such that $|m-\hat{m}|<\epsilon m$. (We restrict our attention to the case of injective functions since this simplifies the exposition of our result. Note that this version with injective functions is still interesting since it is precisely the version needed for estimating the edit distance between two non-repetitive strings~\cite{Andoni+SODA10,LeGall+20,Naumovitz+SODA17}.)

In order to derive worst-case complexity upper bounds that explicitly depend on the number of collisions, we assume that we know a lower bound $\bar m\le m$, and state the complexities using $\bar{m}$.
There is an easy randomized algorithm based on sampling (which has been, for instance, used implicitly in~\cite{Andoni+SODA10,Naumovitz+SODA17}) that outputs with high probability an $\epsilon$-approximation of $m$ using $O\big(\frac{1}{\epsilon}\frac{N}{\sqrt{m}}+\frac{N}{\sqrt{\bar{m}}}\big)$ evaluations of~$f$ and~$g$, which is essentially tight (see Appendix~\ref{sec:classical} for details). Based on the technique of Theorem~\ref{thm:Count}, we design a fast quantum algorithm for this problem:

\begin{theorem}
\label{thm:Coll}
Given a lower bound $\bar{m}>0$ of the number of collision $m$, for any $\epsilon>0$ there is a quantum algorithm that outputs with high probability an estimate~$\hat{m}$ such that $|m-\hat{m}|<\epsilon m$ using
\[\tilde O\left(\frac{1}{\epsilon^{25/24}}\left(\frac{N}{\sqrt{m}}\right)^{2/3}+\left(\frac{N}{\sqrt{\bar{m}}}\right)^{2/3}\right)\]
queries.
\end{theorem}

\subparagraph{Overview of our techniques.}
Let us start by giving an overview of the quantum approximate counting by Brassard, H{\o}yer and Tapp~\cite{Brassard+98}. Roughly speaking, the quantum approximate counting algorithm from~\cite{Brassard+98} applies the technique called quantum phase estimation~\cite{Cleve+98, Kitaev95}, which is based on the quantum Fourier transform, in order to estimate the eigenvalues of the unitary operator used as the main subroutine of Grover's algorithm. This operator, when restricted to the appropriate subspace, is a rotation whose angle is closely related to the number of solutions. The eigenvalues of this operator are thus closely related to the number of solutions as well. A good approximation of the eigenvalues, which can be obtained using quantum phase estimation, then gives a good approximation of the number of solutions.

    The main insight is that the main operator $U$ of quantum walk based search by Magniez, Nayak, Roland and Santha~\cite{MagniezNRS11} \emph{approximately} corresponds to a rotation $\overline{U}$ (see Corollary~\ref{cor:U-Ubar} in Section~\ref{sec:Markov}). To  prove Theorem~\ref{thm:Count}, we observe that the eigenvalues of this rotation $\overline{U}$ are closely related to the fraction of marked states with respect to the stationary distribution of the Markov chain, and then shows that quantum phase estimation applied on $U$ (instead of $\overline{U}$) with good enough precision gives a good estimation of these eigenvalues and thus a good approximation of the number of marked states if the stationary distribution is uniform (Section~\ref{sec:Markov-count}). To prove Theorem~\ref{thm:Coll}, we additionally establish the relation between the number of collisions and the number of marked states in the Markov chain corresponding to a random walk over the Johnson graph (which is the same Markov chain as the one used in Ambainis' element distinctness algorithm~\cite{Ambainis07}), and show that a good approximation for the latter gives a good enough approximation of the former (Section~\ref{sec:coll}).

\section{Preliminaries}
We assume that the reader is familiar with the basics of quantum computation (we refer to, e.g., \cite{QCQI} for a good reference) and describe below phase estimation and quantum walk based search.

\subsection{Phase estimation}\label{sec:PE}
 We first describe phase estimation~\cite{Cleve+98, Kitaev95} (see also~\cite{QCQI}).
Given a unitary matrix $U$ and an eigenvector $|u\rangle$ of $U$ corresponding to an eigenvalue $e^{2\pi i\varphi_u}$, where $\varphi_u\in[0,1)$, phase estimation is a quantum algorithm that outputs a good approximation of $\varphi_u$. The phase estimation algorithm is also given a positive integer $t$ as input, which controls the precision of the estimation. 

In order to implement this algorithm, the unitary $U$ needs to be accessible as follows: 
we have access to a black box that performs a controlled-$U$ operation. 
Note that when an implementation of $U$ as a concrete quantum circuit is available, a circuit implementation of the controlled-$U$ operation can be simply obtained by replacing each elementary gate with its controlled version. 

The phase estimation is described as Algorithm~\ref{algo:PE} below, and a concrete implementation as a quantum circuit which we denote $C_t(U)$ is given in Figure~\ref{fig:circuit}. The total number of controlled-$U$ operations is $2^0+2^1+\cdots+2^{t-1}=O(2^t)$. The following theorem shows that when setting $t$ appropriately, the algorithm outputs a good approximation with high probability.

\begin{theorem}[\cite{QCQI}]\label{th:PE}
For any integer $t_1\ge 1$ and any $\xi\in(0,1]$, when setting $t=t_1+\lceil\log_2(2+\frac{1}{2\xi})\rceil$ the phase estimation algorithm {\bf Est}$(U, |u\rangle, t)$ outputs with probability at least $1-\xi$ an approximation $\tilde{\varphi}_u$ such that $|\varphi_u-\tilde{\varphi}_u|<2^{-t_1}$. 
\end{theorem}

\begin{algorithm}
\caption{The phase estimation algorithm {\bf Est}$(U, |u\rangle, t)$ }
\label{algo:PE}
\DontPrintSemicolon
Prepare the quantum state $|0^t\rangle |u\rangle$.\;
Apply a Hadamard gate of each of the first $t$ qubits.\;
Apply controlled-$U$ operations on the second register, controlled by the first register.\;
Apply the inverse Fourier transform over the first register.\;
Measure the first register in the computational basis and get an integer $b\in\{0,1,\ldots,2^t-1\}$.\;
Output $\frac{b}{2^t}\eqqcolon\tilde{\varphi}_u$.
\end{algorithm}

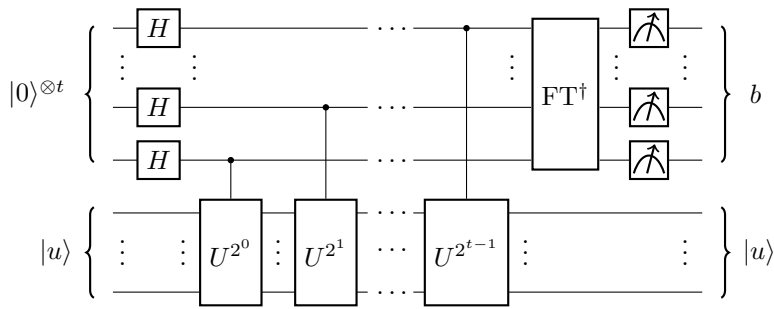
\begin{figure}[ht!]
\centering
\begin{tikzpicture}[scale=0.5,rectnode/.style={thick,shape=rectangle,draw=black,minimum height=14mm, minimum width=8mm},rectnode2/.style={thick,shape=rectangle,draw=black,minimum height=5mm, minimum width=5mm},rectnode3/.style={thick,shape=rectangle,draw=black,minimum height=20mm, minimum width=5mm}]
    \newcommand\YA{7}
        
    \node[rectnode2] (H1) at (-0.3,1.4*\YA) {$H$};
    \node[rectnode2] (H2) at (-0.3,1.1*\YA) {$H$};
    \node[rectnode2] (H3) at (-0.3,0.9*\YA) {$H$};

    \node[rectnode] (U0) at (1.6,0.55*\YA) {$U^{2^0}$};
    \node[rectnode] (U1) at (1.6+2.5,0.55*\YA) {$U^{2^1}$};
    \node[rectnode] (Ut-1) at (7.8,0.55*\YA) {$U^{2^{t-1}}$};
    
    \node[rectnode3] (FT-1) at (10.4,1.15*\YA) {FT$^\dagger$};


     \draw (-1.5,1.4*\YA) -- (-0.85,1.4*\YA);
     \draw (-1.5,1.1*\YA) -- (-0.85,1.1*\YA);
     \draw (-1.5,0.9*\YA) -- (-0.85,0.9*\YA);
     
     \draw (-1.5,0.7*\YA) -- (0.8,0.7*\YA);
     \draw (-1.5,0.4*\YA) -- (0.8,0.4*\YA);

    \draw (2.4,0.7*\YA) -- (3.3,0.7*\YA);
    \draw (2.4,0.4*\YA) -- (3.3,0.4*\YA);

    \draw (4.9,0.7*\YA) -- (5.2,0.7*\YA);
    \draw (4.9,0.4*\YA) -- (5.2,0.4*\YA);
    \node[draw=none,fill=none] at (5.9,0.7*\YA) {$\dots$};
    \node[draw=none,fill=none] at (5.9,0.4*\YA) {$\dots$};
    
    \draw (6.4,0.7*\YA) -- (6.7,0.7*\YA);
    \draw (6.4,0.4*\YA) -- (6.7,0.4*\YA);
     
    \draw (8.9,0.7*\YA) -- (14,0.7*\YA);
    \draw (8.9,0.4*\YA) -- (14,0.4*\YA);
     
    \draw (-1.5+1.75,1.4*\YA) -- (5.2,1.4*\YA);
    \draw (-1.5+1.75,1.1*\YA) -- (5.2,1.1*\YA);
    \draw (-1.5+1.75,0.9*\YA) -- (5.2,0.9*\YA);
    \draw (6.4,1.4*\YA) -- (-0.8+10.3,1.4*\YA);
    \draw (6.4,1.1*\YA) -- (-0.8+10.3,1.1*\YA);
    \draw (6.4,0.9*\YA) -- (-0.8+10.3,0.9*\YA);
    \filldraw[black] (1.6,0.9*\YA) circle (2pt);
    \draw[-] (1.6,0.9*\YA) -- (U0);
    \filldraw[black] (1.6+2.5,1.1*\YA) circle (2pt);
    \draw[-] (1.6+2.5,1.1*\YA) -- (U1);
    \filldraw[black] (7.8,1.4*\YA) circle (2pt);
    \draw[-] (7.8,1.4*\YA) -- (Ut-1);
    \node[draw=none,fill=none] at (5.9,1.4*\YA) {$\dots$};
    \node[draw=none,fill=none] at (5.9,1.1*\YA) {$\dots$};
    \node[draw=none,fill=none] at (5.9,0.9*\YA) {$\dots$};
     
    \draw (11.3,1.4*\YA) -- (12.1,1.4*\YA);
    \draw (11.3,1.1*\YA) -- (12.1,1.1*\YA);
    \draw (11.3,0.9*\YA) -- (12.1,0.9*\YA);
     
    \draw[thick,decorate,decoration={brace,amplitude=1mm}] (-2,-1.8+1.15*\YA) -- (-2,1.8+1.15*\YA);
    \node[draw=none,fill=none] at (-3.5,1.15*\YA) {$\ket{0}^{\otimes t}$};
    \node[draw=none,fill=none] at (-1.25,1.28*\YA) {$\vdots$};
    \node[draw=none,fill=none] at (0.65,1.28*\YA) {$\vdots$};
    \node[draw=none,fill=none] at (9,1.28*\YA) {$\vdots$};
    \node[draw=none,fill=none] at (11.75,1.28*\YA) {$\vdots$};
     
    \draw[thick,decorate,decoration={brace,amplitude=1mm}] (-2,-1.2+0.55*\YA) -- (-2,1.2+0.55*\YA);
    \node[draw=none,fill=none] at (-3,0.55*\YA) {$\ket{u}$}; 
    \node[draw=none,fill=none] at (-1.25,0.58*\YA) {$\vdots$};
    \node[draw=none,fill=none] at (0.35,0.58*\YA) {$\vdots$};
    \node[draw=none,fill=none] at (2.85,0.58*\YA) {$\vdots$};
    \node[draw=none,fill=none] at (5.9,0.56*\YA) {$\dots$};
    \node[draw=none,fill=none] at (9.35,0.58*\YA) {$\vdots$};
    \node[draw=none,fill=none] at (13.55,0.58*\YA) {$\vdots$};

    \draw[thick,decorate,decoration={brace,mirror,amplitude=1mm}] (14.5,-1.2+0.55*\YA) -- (14.5,1.2+0.55*\YA);
    \node[draw=none,fill=none] at (15.5,0.55*\YA) {$\ket{u}$};

    \draw[thick,decorate,decoration={brace,mirror,amplitude=1mm}] (14.5,-1.8+1.15*\YA) -- (14.5,1.8+1.15*\YA); 
     \node[rectnode2] (M1) at (12.6,1.4*\YA) {};
     \node[rectnode2] (M2) at (12.6,1.1*\YA) {};
     \node[rectnode2] (M3) at (12.6,0.9*\YA) {};
     \draw (13.1,1.4*\YA) -- (14,1.4*\YA);
     \draw (13.1,1.1*\YA) -- (14,1.1*\YA);
     \draw (13.1,0.9*\YA) -- (14,0.9*\YA);
     \node[draw=none,fill=none] at (13.55,1.28*\YA) {$\vdots$};
     \node[draw=none,fill=none] at (15.4,1.15*\YA) {$b$};

     \draw[->, thick] (12.55,1.35*\YA) -- (12.75,1.46*\YA);
     \draw[thick] (12.2,1.35*\YA) .. controls  (12.4,1.44*\YA) and  (12.8,1.44*\YA) .. (13,1.35*\YA);
     \draw[->, thick] (12.55,1.05*\YA) -- (12.75,1.16*\YA);
     \draw[thick] (12.2,1.05*\YA) .. controls  (12.4,1.14*\YA) and  (12.8,1.14*\YA) .. (13,1.05*\YA);
     \draw[->, thick] (12.55,0.85*\YA) -- (12.75,0.96*\YA);
     \draw[thick] (12.2,0.85*\YA) .. controls  (12.4,0.94*\YA) and  (12.8,0.94*\YA) .. (13,0.85*\YA);
\end{tikzpicture}
\vspace*{13pt}
\caption{\label{fig:circuit}The circuit $C_t(U)$ for phase estimation, where $H$ is the Hadamard gate and FT$^\dagger$ is the inverse Fourier transform.}
\end{figure}

\subsection{Search via quantum walk}\label{sec:Markov}
In this subsection we describe the quantum walk based search approach  (also called search via quantum walk) by Magniez, Nayak, Roland and Santha~\cite{MagniezNRS11}.

\subparagraph{Classical Markov chains.}
A Markov chain $P$ with a finite state space $X$ corresponds to a random walk over a weighted graph. In this paper, as in~\cite{MagniezNRS11}, we write $|X|=n$ and identify the Markov chain with its transition matrix $P=(p_{xy})$, where $p_{xy}$ is the probability of transition from $x$ to $y$ (the matrix $P$ is a stochastic matrix). The Markov chain is ergodic if its reachability graph is connected and non-bipartite. An ergodic Markov chain has a unique stationary distribution, which we denote $\pi$. The Markov chain is reversible if the equality $\pi_xp_{xy}=\pi_yp_{yx}$ holds for all $x,y\in X$ (for example, any Markov chain induced from an undirected weighted graph is reversible). 

All the Markov chains considered in this paper are ergodic and reversible. Let $\mathcal{M}$ be a non-empty subset of~$X$. The elements of $\mathcal{M}$ are called the marked states of~$P$.   Define 
\[
p_M=\sum_{x\in\mathcal{M}}\pi_x,
\]
the fraction of marked states in the stationary distribution. 

\subparagraph{Quantum implementation of Markov chains.} Due to reversibility issues in the quantum setting, the approach by Magniez, Nayak, Roland and Santha~\cite{MagniezNRS11} for implementing Markov chains considers pairs of states (i.e., the space $X\times X$) instead of working with states (i.e., the space $X$). Consider the vector space $\mathcal{H}=\mathbb{C}^{X\times X}=\operatorname{span}(\ket{x}\ket{y} : (x,y)\in X\times X)$.
Let $\mathcal{A}$ and $\mathcal{B}$ be the subspaces of $\mathcal{H}$ defined as $\mathcal{A}=\operatorname{span}(|x\rangle |p_x\rangle : x\in X)$ and $\mathcal{B}=\operatorname{span}(|p^*_y\rangle |y\rangle : y\in X)$, where 
\[
|p_x\rangle =\sum\limits_{y\in X}\sqrt{p_{xy}}|y\rangle \hspace{3mm}\textrm{ and }\hspace{3mm} |p_y^*\rangle =\sum\limits_{x\in X}\sqrt{p_{yx}}|x\rangle.
\]
The unitary operation $W(P)$ defined on $\mathcal{H}$ by $W(P)=\operatorname{ref}(\mathcal{B})\cdot\operatorname{ref}(\mathcal{A})$ is called the quantum walk based on the classical chain $P$. 


Define the following two quantum states:
\begin{align*}
|\pi\rangle &=\sum\limits_{x\in X}\sqrt{\pi_x}|x\rangle |p_x\rangle,\\ 
|\mu\rangle &=\frac{1}{\sqrt{p_M}}\sum\limits_{x\in\mathcal{M}}\sqrt{\pi_x}|x\rangle |p_x\rangle, 
\end{align*}
and write $\mathcal{S}=\operatorname{span}(|\pi\rangle , |\mu\rangle)$. Let $|\mu^\perp\rangle$ denote the state orthonormal to $|\mu\rangle$ in the subspace~$\mathcal{S}$ of $\mathcal{H}$. Then
\[|\pi\rangle =\sin\varphi |\mu\rangle +\cos\varphi |\mu^\perp\rangle,\]
where $\varphi\in [0,\frac{\pi}{2})$, $\sin\varphi =\sqrt{p_M}$ (remember that $p_M$ denotes the fraction of marked items of~$P$).
Let us write $\operatorname{ref}(\pi)$ the unitary operation defined on $\mathcal{A}+\mathcal{B}$ such that $\operatorname{ref}(\pi)|\pi\rangle=|\pi\rangle$ and $\operatorname{ref}(\pi)|\psi\rangle=-|\psi\rangle $ for any $|\psi\rangle\in\mathcal{A}+\mathcal{B}$ orthogonal to $\ket{\pi}$. Define the unitary operation $V_0$ on $\mathcal{A}+\mathcal{B}$ as follows: 
\[V_0\colon |x\rangle |y\rangle \mapsto\left \{
\begin{aligned}
-|x\rangle |y\rangle  & \quad\text{if }x\in\mathcal{M}\\
|x\rangle |y\rangle  & \quad\text{otherwise}
\end{aligned}
\right.\]
and let $\overline{V}=\operatorname{ref}(\pi)\cdot V_0$. In the subspace~$\mathcal{S}$ the unitary operation $\overline{V}$ is a rotation of angle $2\varphi$.

One of the main technical contributions of~\cite{MagniezNRS11} is the construction of a quantum circuit which efficiently implements an operator that is close to the reflection $\operatorname{ref}(\pi)$.

\def\thm6inQW{\cite[Section 3.2]{MagniezNRS11}}

\begin{theorem}[\thm6inQW]
\label{thm:QW}
Let $P$ be a reversible ergodic Markov chain with a state space of size $n\geq 2$. Let $\delta$ denote the spectral gap of $P$.
Then for any integer $k$ there exists a quantum circuit $R(P)$ that acts on $O(\log n)+\tau$ qubits, where $\tau=O(k\log(1/\delta))$, and satisfies the following properties:
\begin{enumerate}
\item The circuit $R(P)$ uses $O(k\log(1/\delta))$ Hadamard gates, $O(k\log^2(1/\delta))$ controlled phase rotations, and makes at most $O(k/\sqrt{\delta})$ calls to the controlled quantum walk c-$W(P)$ and its inverse c-$W(P)^\dagger$.
\item For any $|\psi\rangle\in\mathcal{A}+\mathcal{B}$, $\Vert (R(P)-\operatorname{ref}(\pi)\otimes I)|\psi\rangle |0^{\tau}\rangle\Vert\leq 2^{1-k}$.
\end{enumerate}
\end{theorem}

The following immediate corollary of Theorem~\ref{thm:QW} gives an upper bound for the difference between $U=R(P)\cdot (V_0\otimes I)$ and the rotation $\overline{U}=\overline{V}\otimes I$.

\begin{corollary}
\label{cor:U-Ubar}
For any $|\psi\rangle\in\mathcal{A}+\mathcal{B}$, $\Vert (U-\overline{U})|\psi\rangle |0^{\tau}\rangle\Vert\leq 2^{1-k}$.
\end{corollary}

\subparagraph{Quantum walk based search.} When considering quantum walk based search the input is given as a black box, as for Grover's search. The set $M$ of marked states is defined using the input in such a way that finding a marked state corresponds to finding a solution to the search problem. Additionally, some data from the input is stored in a data structure~$d$. Mathematically, this corresponds to working with $X_d=\{ (x, d(x)): x\in X\}$ and $\mathcal{H}_d=\mathbb{C}^{X_d\times X_d}$ instead of $X$ and~$\mathcal{H}$. For convenience we simply write $\ket{x}_d$ instead of $\ket{(x, d(x))}$, for any $x\in X$. The query complexity of the quantum walk based search can then be expressed in terms of three costs (see \cite[Section 1.3]{MagniezNRS11} for details):
\begin{itemize}
\item Setup cost $\mathsf{S}$: The cost of constructing the state $\sum_x\sqrt{\pi_x}|x\rangle _d|\bar{0}\rangle _d$ from $|\bar{0}\rangle _d|\bar{0}\rangle _d$.
\item Update cost $\mathsf{U}$: The cost of realizing any of the unitary transformations
\begin{align*}
|x\rangle _d|\bar{0}\rangle _d\ &\mapsto\ |x\rangle _d\sum_y\sqrt{p_{xy}}|y\rangle _d\\
|\bar{0}\rangle _d|y\rangle _d\ &\mapsto\ \sum_x\sqrt{p_{xy}^*}|x\rangle _d|y\rangle _d
\end{align*}
and their inverses.
\item Checking cost $\mathsf{C}$: The cost of realizing the following conditional phase flip
\[|x\rangle _d|y\rangle _d\ \mapsto\ \left \{
\begin{aligned}
-|x\rangle _d|y\rangle _d & \quad\text{if }x\in\mathcal{M}.\\
|x\rangle _d|y\rangle _d & \quad\text{otherwise}.
\end{aligned}
\right.\]
\end{itemize}

A single application of $R(P)$ can be implemented using $O(\frac{k}{\sqrt{\delta}}\mathsf{U})$ queries, from Theorem~\ref{thm:QW}.   
A single application of $U$ can thus be implemented using $O(\frac{k}{\sqrt{\delta}}\mathsf{U}+\mathsf{C})$ queries. Magniez, Nayak, Roland and Santha~\cite{MagniezNRS11} showed that by first preparing the initial state using $\mathsf{S}$ queries, and then applying $O(1/\sqrt{\LBDofFraction})$ times the operation $U$, a marked state (and thus a solution to the search problem) is obtained with high probability. The overall complexity of this strategy is 
\[
O\left(\mathsf{S}+\frac{1}{\sqrt{\LBDofFraction}}\left(\frac{k}{\sqrt{\delta}}\mathsf{U}+\mathsf{C}\right)\right)
\]
queries.
The factor $k$ can be removed from the complexity using more work (see \cite[Section~4]{MagniezNRS11}), which gives the complexity stated in Eq.~(\ref{eq1}) in the introduction.

\section{Quantum Counting for Markov Chains}\label{sec:Markov-count}
In this section we prove Theorem~\ref{thm:Count}: for a reversible ergodic Markov chain which has the uniform stationary distribution and a finite set~$X$ of size $|X|=N$ and a marked subset $\mathcal{M}\subseteq X$, the goal is to find an estimation $\hat{M}$ of $M=|\mathcal{M}|$ such that $|M-\hat{M}|<\epsilon M$. The main idea is to apply phase estimation on the operator $U$ introduced in Section~\ref{sec:Markov}. Before giving the proof of Theorem~\ref{thm:Count} we present a high-level description of our strategy and prove a technical result (Lemma~\ref{lemma:difference}).

Consider first the operator $\overline{U}$ introduced in Section~\ref{sec:Markov}. Since $\overline{U}$ is a rotation of angle $2\varphi$ on $\mathcal{S}\otimes\ket{0^\tau}$, in this subspace it has orthonormal eigenvectors
\[\ket{\Psi_{\pm}}=\frac{1}{\sqrt{2}}(\ket{\mu}\pm i\ket{\mu^\perp})\otimes\ket{0^\tau}\]
corresponding to the eigenvalues $e^{\pm i2\varphi}$.
Applying phase estimation (Algorithm~\ref{algo:PE}) with $\overline{U}$ on $\ket{\Psi_\pm}$ would give an approximation of $\pm\varphi /\pi$ with high probability,\footnote{Observe that we have a $\pi$ factor here since the version of phase estimation presented in Section~\ref{sec:PE} assumes that the eigenvalues are of the form $e^{2\pi i \varphi_u}$ with $\varphi_u\in[0,1)$.}
\, from which we would be able to obtain a good approximation of $M$ since $\sin(\varphi)=\sqrt{p_M}$, where $\sqrt{p_M}=\sqrt{M/N}$ if the stationary distribution of $P$ is uniform. There are nevertheless two issues with this strategy: we do not know the eigenvectors $\ket{\Psi_{\pm}}$ and we cannot apply $\overline{U}$ directly.

The first issue is dealt with in a fairly standard way: we will instead apply phase estimation using the state $\ket{\pi}\ket{0^\tau}\in\mathcal{S}\otimes\ket{0^\tau}$. This state, which we know how to construct (at cost $\mathsf{S}$ queries), can be expressed as
\begin{equation}\label{eq:decomp}
\ket{\pi}\ket{0^\tau}=\sin\varphi |\mu\rangle\ket{0^\tau} +\cos\varphi |\mu^\perp\rangle\ket{0^\tau}=-\frac{i}{\sqrt{2}}e^{i2\varphi}\ket{\Psi_+}+\frac{i}{\sqrt{2}}e^{-i2\varphi}\ket{\Psi_-},
\end{equation}
and thus applying Algorithm~\ref{algo:PE} with $\overline{U}$ on $\ket{\pi}\ket{0^\tau}$ would output with high probability an approximation of either $\varphi /\pi$ or $-\varphi /\pi$, which would be enough to obtain an approximation of $p_M$.

To solve the second problem, we will simply apply phase estimation on $U$ instead of~$\overline{U}$. As in Section~\ref{sec:PE}, let us use $C_t(U)$ to denote the circuit of phase estimation applied on $U$ with the parameter $t$.
We can upper bound the difference between $C_t(U)$ and $C_t(\overline{U})$ operating on $\ket{0^t}\ket{\psi}\ket{0^{\tau}}$ with $\ket{\psi}\in\mathcal{A}+\mathcal{B}$ as follows.

\begin{lemma}
\label{lemma:difference}
For any $|\psi\rangle\in\mathcal{A}+\mathcal{B}$, $\Vert (C_t(U)-C_t(\overline{U}))|0^t\rangle |\psi\rangle |0^{\tau}\rangle\Vert\leq 2^{2t-k+1}$.
\end{lemma}

\begin{proof}
Let $|\psi\rangle\in\mathcal{A}+\mathcal{B}$.
Since $U$ and $\overline{U}$ are unitary, Corollary~\ref{cor:U-Ubar} implies that for any $l\in\mathbb{N}$,
\begin{align*}
    \Vert (U^l-\overline{U}^l)|\psi\rangle |0^{\tau}\rangle\Vert &=\Vert (U-\overline{U})(U^{l-1}+U^{l-2}\overline{U}+\cdots +\overline{U}^{l-1})|\psi\rangle |0^{\tau}\rangle\Vert\\
    &\leq l\Vert (U-\overline{U})|\psi\rangle |0^{\tau}\rangle\Vert\\
    &\leq l\cdot 2^{1-k}.
\end{align*}
The quantum state at the end of Step 3 of Algorithm~\ref{algo:PE} is 
\[
\frac{1}{\sqrt{2^t}}\sum_{j=0}^{2^t-1}|j\rangle U^j |\psi\rangle |0^{\tau}\rangle.
\]
Thus we have
\begin{align*}
    \Vert (C_t(U)-C_t(\overline{U}))|0^t\rangle |\psi\rangle |0^{\tau}\rangle\Vert &\leq 2^t\max_{0\leq l\leq 2^t-1}\Vert(U^{l}-\overline{U}^{l})|\psi\rangle |0^{\tau}\rangle\Vert\\
    &\leq 2^t(2^{t}-1)2^{1-k}\\
    &\leq 2^{2t-k+1},
\end{align*}
as claimed.
\end{proof}

We are now ready to give a proof of Theorem~\ref{thm:Count} by putting all these ingredients together, adjusting all the parameters and making the analysis of the approximation errors.  

\begin{proof}[Proof of Theorem~\ref{thm:Count}.]
We divide the proof into several components.

\subparagraph*{Analysis of phase estimation on the ideal rotation $\bold{\overline{U}}$.} 
From Theorem~\ref{th:PE} we know that when choosing 
parameters $t_1=\left\lceil \log_2\left(\frac{5\pi}{\epsilon\sqrt{\LBDofFraction}}\right)-1\right\rceil$, $\xi =0.01$, $t=t_1+\lceil\log_2(2+\frac{1}{2\xi})\rceil$, the algorithm {\bf Est}($\overline{U}$, $\ket{\Psi_{\pm}} |0^{\tau}\rangle$, $t$) would output a value $z$ such that $|z -(\pm\varphi /\pi)|<2^{-t_1}$ with probability at least $0.99$. 

Let us now interpret these probabilities in a more algebraic way. Let $M_+$ and $M_-$ be the projection over the subspaces $\operatorname{span}\left\{\ket{\tilde{b}}: \left|\frac{\tilde{b}}{2^t}-\frac{\varphi}{\pi}\right|<2^{-t_1}\right\}$ and $\operatorname{span}\left\{\ket{\tilde{b}}: \left|\frac{\tilde{b}}{2^t} + \frac{\varphi}{\pi}\right|<2^{-t_1}\right\}$, respectively, of the vector space $\mathbb{C}^{2^t}$. We thus have:
\begin{align*}
\Vert (M_+ \otimes I)C_t(\overline{U})\ket{0^t}\ket{\Psi_{+}}\ket{0^{\tau}}\Vert^2=\Pr\left[\left|z -\frac{\varphi}{\pi}\right|<2^{-t_1}\right]&\ge 0.99\\
\Vert (M_- \otimes I)C_t(\overline{U})\ket{0^t}\ket{\Psi_{-}}\ket{0^{\tau}}\Vert^2=\Pr\left[\left|z +\frac{\varphi}{\pi}\right|<2^{-t_1}\right]&\ge 0.99.
\end{align*}

Let us now analyze the output of {\bf Est}($\overline{U}$,  $|\pi\rangle |0^{\tau}\rangle$, $t$). 
For conciseness, we write $\ket{\overline{\Phi}}=C_t(\overline{U})\ket{0^t}\ket{\pi}\ket{0^{\tau}}$.
Remember the decomposition of Eq. (\ref{eq:decomp}) and observe that ${\big\Vert -\frac{i}{\sqrt{2}}e^{i2\varphi}\big\Vert^2}=\big\Vert \frac{i}{\sqrt{2}}e^{-i2\varphi}\big\Vert^2=\frac{1}{2}$. Thus we have:
\begin{equation}\label{eq:bound}
\Vert (M_+ \otimes I)\ket{\overline{\Phi}}\Vert^2\ge 0.99/2
\:\:\textrm{ and }\:\:
\Vert (M_- \otimes I)\ket{\overline{\Phi}}\Vert^2\ge 0.99/2,
\end{equation}
which means that {\bf Est}($\overline{U}$,  $|\pi\rangle |0^{\tau}\rangle$, $t$) outputs a value $z$ such that $|z -\varphi/\pi|<2^{-t_1}$ with probability at least $0.99/2$ and $|z +\varphi/\pi|<2^{-t_1}$ with probability at least $0.99/2$.

\subparagraph*{Description of the algorithm.}
We apply phase estimation on $U$ instead of $\overline{U}$, with the input state $|\pi\rangle |0^{\tau}\rangle$, where $U$ is the operator introduced in Section~\ref{sec:Markov}. Remember that $U$ is defined using a parameter $k$ (see the statement of Theorem~\ref{thm:QW}). We set $k=2t+t_1+1$.

The algorithm for Theorem~\ref{thm:Count} is as follows.

\begin{algorithm}[H]
\caption{Quantum Approximate Counting for Markov Chains}
\label{algo:count}
\SetKwInOut{Input}{Input}
\setcounter{AlgoLine}{-1}
\DontPrintSemicolon
\Input{desired accuracy parameter $\epsilon >0$, lower bound on the fraction of marked states $\LBDofFraction$ of the Markov chain.}
Set $t=\left\lceil \log_2\left(\frac{5\pi}{\epsilon\sqrt{\LBDofFraction}}\right)-1\right\rceil +\lceil\log_2(2+\frac{1}{0.02})\rceil$ and $k=2t+\left\lceil \log_2\left(\frac{5\pi}{\epsilon\sqrt{\LBDofFraction}}\right)-1\right\rceil+1$.\;
Apply {\bf Est}($U$, $|\pi\rangle |0^{\tau}\rangle$, $t$)  and denote $y$ its output.\;
Output $\hat{M}=N\sin^2(y\pi)$.
\end{algorithm}

\subparagraph*{Error Analysis.}
We are going to show that $|M-\hat{M}|<\epsilon$, where $M=N\sin^2\varphi$ and $\hat{M}=N\sin^2(y\pi)$, by analyzing the quantity $\bigl|\varphi -|y|\pi\bigr|$.

Let us write $\ket{\Phi}=C_t(U)\ket{0^t}\ket{\pi}\ket{0^{\tau}}$.
For any measurement operator $M_m\in\{M_+, M_-\}$, we have
\begin{align*}
\Vert (M_m \otimes I)\ket{\Phi}\Vert^2 &=
\Vert (M_m \otimes I)\ket{\overline\Phi} +(M_m \otimes I)(\ket{\Phi}-\ket{\overline{\Phi}})\Vert^2\\
&\le
\Vert (M_m \otimes I)\ket{\overline\Phi}\Vert^2 +\Vert (M_m \otimes I)(\ket{\Phi}-\ket{\overline{\Phi}})\Vert^2\\
&\le
\Vert (M_m \otimes I)\ket{\overline\Phi}\Vert^2 +\Vert \ket{\Phi}-\ket{\overline{\Phi}}\Vert^2\\
&\le
\Vert (M_m \otimes I)\ket{\overline\Phi}\Vert^2 +2^{-2t_1},
\end{align*}
where we used Lemma~\ref{lemma:difference} to derive the last inequality (remember that we are taking $k=2t+t_1+1$).
Combined with (\ref{eq:bound}), this implies that with probability at least $0.99/2-2^{-2t_1}$ the output of {\bf Est}($U$, $|\pi\rangle |0^{\tau}\rangle$, $t$) satisfies the condition $|y-\varphi /\pi|\leq 2^{-t_1}$, and with probability at least $0.99/2-2^{-2t_1}$ it satisfies the condition $|y+\varphi /\pi|\leq 2^{-t_1}$.
It follows from this condition that $\big| |y|-\varphi /\pi\big|\leq 2^{-t_1}$ holds with probability at least $0.99-2^{1-2t_1}$, since $\varphi\geq 0$. Assume that this inequality holds.
Then $\bigl|\varphi -|y|\pi\bigr|=\pi\left|\frac{\varphi}{\pi}-|y|\right|\leq 2^{-t_1}\pi$.
Notice that for any $\alpha, \beta\in\mathbb{R}$, $|\sin{\alpha}-\sin{\beta}|\leq |\alpha-\beta|$ by the mean value theorem.
We thus have
\begin{align*}
\left|\sqrt{p_M}-\sin{(|y|\pi)}\right|&=\big|\sin{\varphi}-\sin{(|y|\pi)}\big|\\
&\leq\bigl|\varphi -|y|\pi\bigr|\\
&<2^{-t_1}\pi\\
&\leq\frac{2}{5}\epsilon\sqrt{\LBDofFraction}
\end{align*}
and it follows that
\begin{align*}
|\hat{M}-M| &=\left|N\sin^2(|y|\pi)-Np_M\right|\\
&=N\left|(\sqrt{p_M}-\sin (|y|\pi))^2+2\sqrt{p_M}\sin (|y|\pi)-2p_M\right|\\
&\leq N\left(\left|\sqrt{p_M}-\sin (|y|\pi)\right|^2+2\sqrt{p_M}\left|\sqrt{p_M}-\sin (|y|\pi)\right|\right)\\
&<\left(\left(\frac{2}{5}\epsilon\sqrt{\LBDofFraction}\right)^2+2\sqrt{p_M}\cdot\frac{2}{5}\epsilon\sqrt{\LBDofFraction}\right)N\\
&<\frac{24}{25}\epsilon p_MN\\
&<\epsilon M.
\end{align*}

\subparagraph*{Complexity.}
Under the framework presented in Section~\ref{sec:Markov}, the quantum state $\ket{\pi}$ can be created using $\mathsf{S}$ queries, and the operator $U$ can be implemented using $\frac{k}{\sqrt{\delta}}\mathsf{U}+\mathsf{C}$ queries.
The operator $U$ is applied for $O(2^t)=O(\frac{1}{\epsilon}\frac{1}{\sqrt{\LBDofFraction}})$ times by the phase estimation algorithm.
The overall complexity of Algorithm~\ref{algo:count} is thus 
\[
O\left(
\mathsf{S}+\frac{1}{\epsilon}\frac{1}{\sqrt{\LBDofFraction}}\left(\frac{k}{\sqrt{\delta}}\mathsf{U}+\mathsf{C}\right)
\right)
=
\tilde O\left(
\mathsf{S}+\frac{1}{\epsilon}\frac{1}{\sqrt{\LBDofFraction}}\left(\frac{1}{\sqrt{\delta}}\mathsf{U}+\mathsf{C}\right)
\right)
\]
queries, as claimed.
\end{proof}

From the argument above, we then obtain the following immediate corollary for any reversible ergodic Markov chain whose stationary distribution is not necessarily uniform.

\begin{corollary}
Let $P$ be a reversible ergodic Markov chain, $\delta>0$ be the spectral gap of $P$ and $\LBDofFraction>0$ be a known lower bound on the fraction $p_M$ of marked states of $P$ with respective to its stationary distribution. For any $\epsilon\in (0, 1)$, there exists a quantum algorithm which outputs with high probability an estimate $\hat{p}_M$ such that $|p_M-\hat{p}_M|<\epsilon p_M$, and uses 
\[
\tilde O\left(\mathsf{S}+\frac{1}{\epsilon}\frac{1}{\sqrt{\LBDofFraction}}\left(\frac{1}{\sqrt{\delta}}\mathsf{U}+\mathsf{C}\right)\right)
\] 
queries to the black box.
\end{corollary}

\section{Application: Approximate Collision Counting}\label{sec:coll}

Consider two injective functions $f, g:\{ 1, 2,\dots , N\}\longrightarrow\{ 1, 2,\dots , K\}$, where $K> N$, given as quantum oracles $U_f,U_g$. More precisely, $U_f$ and $U_g$ act on $\lceil\log_2(N)\rceil+\lceil\log_2(K)\rceil$ qubits and are defined as follows: for any $i\in \{ 1, 2,\dots , N\}$ and any $a\in\{1, 2,\dots , K\}$,
\begin{align*}
U_f\colon \ket{i}\ket{a}&\mapsto \ket{i}\ket{a\oplus f(i)},\\
U_g\colon \ket{i}\ket{a}&\mapsto \ket{i}\ket{a\oplus g(i)},
\end{align*}
where we are interpreting integers as binary strings (via some fixed encoding), and $\oplus$ denotes the bit-wise parity. Let $m$ be the number of collisions, i.e., $m=|\{(i, j)\in\{1,\dots , N\}\times\{1,\dots , N\}:f(i)=g(j)\}|$.
In this section we show how to use Theorem~\ref{thm:Count} to approximate $m$.

Define $D$ to be the disjoint union of the domain of $f$ and the domain of $g$.
Note that $|D|=2N$.
For convenience, we will later identify $D$ with the set $\{1,\dots , N\}\times\{0, 1\}$.
Similarly to Ambainis's algorithm for element distinctness~\cite{Ambainis07}, we consider a Johnson graph $G_r$ determined by a parameter $r\in\{1, 2,\dots ,2N\}$, which is defined by $V(G_r)=\{ S\subset D: |S|=r\}$ and $E(G_r)=\{ST: |S\cup T|=r+1\}$.
Then $|V(G_r)|=\binom{2N}{r}$.
We say a vertex $S$ in $G_r$ is marked if it contains a collision, i.e., there exist $(i, 0), (j, 1)\in S$ such that $f(i)=g(j)$.
Denote the Markov chain corresponding to the random walk on $G_r$ by $P_{G, r}$.
We identify the states of $P_{G, r}$ with the vertices of $G_r$.

By the principle of inclusion and exclusion, the number of marked vertices in $G_r$ is
\[M_r=\sum_{j=1}^{m}(-1)^{j+1}\binom{m}{j}\binom{2N-2j}{r-2j}=m\binom{2N-2}{r-2}+\sum_{j=2}^{m}(-1)^{j+1}\binom{m}{j}\binom{2N-2j}{r-2j}.\]

Define $\ell_j\coloneqq\binom{m}{j}\binom{2N-2j}{r-2j}$, $R_j\coloneqq\frac{\ell_{j+1}}{\ell_j}$ and $R_0\coloneqq 1$.
Then
\[M_r=\sum_{j=1}^{m}(-1)^{j+1}\ell_j=\ell_1\sum_{i=0}^{m-1}\left((-1)^{i}\displaystyle\prod_{j=0}^iR_j\right).\]
Observe that $R_j\geq R_{j+1}$ for ${j=0,\dots ,m-2}$.

The following lemma characterizes the relation between an approximation of $M_r$ and an approximation of $m$.

\begin{lemma}
\label{lemma:coll}
Let $\epsilon\in(0,1]$ and $r\in\{ 2, 3,\dots , 2N\}$ be such that $R_1=\frac{m-1}{2}\frac{(r-2)(r-3)}{(2N-2)(2N-3)}<\sqrt{\frac{\epsilon}{2}}$ holds.
Define $R^\prime =\frac{m_{\text{up}}-1}{2}\frac{(r-2)(r-3)}{(2N-2)(2N-3)}$ and assume that $R^\prime\leq\sqrt{\epsilon}$, where $m_{\text{up}}$ is an upper bound of $m$.
Then for any value $\hat{M}_r$ that satisfies the condition $|M_r-\hat{M}_r|<\frac{\epsilon}{3}M_r$, the inequality \[\left|m-\frac{1+R^\prime}{\binom{2N-2}{r-2}}\hat{M}_r\right|<\epsilon m\]
holds.
\end{lemma}

\begin{proof}[Proof of Lemma~\ref{lemma:coll}.]
From the definitions of $R_j$ and $R^\prime$, we have $R_{j+1}\leq R_j\leq 1$ for ${j=0,\dots , m-2}$ and $R_1\leq R^\prime$. Then
\begin{align*}
\left|m-\frac{1+R^\prime}{\binom{2N-2}{r-2}}M_r\right| & =\left|m-\frac{1+R^\prime}{\binom{2N-2}{r-2}}\ell_1\sum_{i=0}^{m-1}\left((-1)^{i}\prod_{j=0}^iR_j\right)\right|\\
& =m\left|1-(1+R^\prime)(1-R_1+R_1R_2-\cdots)\right|\\
& \leq m\left|R_1^2-(1+R^\prime)R_1R_2(1-R_3+\cdots)\right|\\
& \leq mR_1^2\\
& \leq \frac{\epsilon}{2}m.
\end{align*}
Hence $\frac{1+R^\prime}{\binom{2N-2}{r-2}}M_r\leq m\left( 1+\frac{\epsilon}{2}\right)$, and
\begin{align*}
\left|m-\frac{1+R^\prime}{\binom{2N-2}{r-2}}\hat{M}_r\right| & \leq\left| m-\frac{1+R^\prime}{\binom{2N-2}{r-2}}M_r\right| +\frac{1+R^\prime}{\binom{2N-2}{r-2}}\left|M_r-\hat{M}_r\right|\\
& <\frac{\epsilon}{2}m+\frac{1+R^\prime}{\binom{2N-2}{r-2}}\frac{\epsilon}{3}M_r\\
& \leq\frac{\epsilon}{2}m+m(1+\frac{\epsilon}{2})\frac{\epsilon}{3}\\
& \leq\epsilon m(\frac{1}{2}+\frac{1}{3}+\frac{1}{6})\\
& =\epsilon m,
\end{align*}
as claimed.
\end{proof}

We are now ready to give the proof of Theorem~\ref{thm:Coll}.

\begin{proof}[Proof of Theorem~\ref{thm:Coll}.]
We first compute a constant-factor approximation of the number of collision $m$.
Concretely, we can use Algorithm~\ref{algo:rough} in Appendix~\ref{sec:rough} with the lower bound $\bar{m}$ 
to obtain a value $\hat{m}_1$ such that $|m-\hat{m}_1|<m/2$, i.e., such that $m/2<\hat{m}_1<3m/2$ or, equivalently, $2\hat{m}_1/3<m<2\hat{m}_1$.
Denote $\left\lfloor2\hat{m}_1/3\right\rfloor$ and $\left\lceil 2\hat{m}_1\right\rceil$ by $m_{\text{low}}$ and $m_{\text{up}}$, respectively.

Let $r$ be a positive integer satisfying the conditions $R_1<\sqrt{\frac{\epsilon}{2}}$ and $R^\prime\leq\sqrt{\epsilon}$ (the choice of~$r$ will be discussed later).
Take
\[\LBDofFraction =\frac{1}{|V(G_{r})|}\sum_{j=1}^{m_{\text{low}}}(-1)^{j+1}\binom{m_{\text{low}}}{j}\binom{2N-2j}{r-2j}\leq p_{M, r}.\]
Since $G_r$ is a regular graph, the stationary distribution of $P_{G, r}$ is uniform. It follows that applying Theorem~\ref{thm:Count} on $P_{G, r}$ with $\epsilon /3$ and $\LBDofFraction$ outputs with high probability a value $\hat{M}_{r}$ such that ${|M_{r}-\hat{M}_{r}|<\frac{\epsilon}{3}M_{r}}$ using 
\[
\tilde{O}\left(\mathsf{S}+\frac{1}{\epsilon}\frac{1}{\sqrt{\LBDofFraction}}\left(\frac{1}{\sqrt{\delta}}\mathsf{U}+\mathsf{C}\right)\right)
\]
queries.
Then we output $\hat{m}=\frac{1+R^\prime}{\binom{2N-2}{r-2}}\hat{M}_{r}$.
Lemma~\ref{lemma:coll} guarantees that $|m-\hat{m}|<\epsilon m$ when ${|M_{r}-\hat{M}_{r}|<\frac{\epsilon}{3}M_{r}}$ holds.

\subparagraph*{Complexity.}
The spectral gap $\delta_r$ of the Johnson graph $G_r$ is such that $\delta_r=\Theta(1/r)$ (see for instance~\cite{Ambainis07, MagniezNRS11}).

Applying Algorithm~\ref{algo:rough} costs $\tilde O\left(\left(N/\sqrt{\bar{m}}\right)^{2/3}\right)$ queries (see Theorem~\ref{th:rough}).

In the main algorithm, the setup cost is $\mathsf{S}=r$, the update cost $\mathsf{U}=2$, and the checking cost $\mathsf{C}=0$.
Since $\hat{m}_1=\Theta(m)$ and we chose $\LBDofFraction$ as
\[
\LBDofFraction =m_{\text{low}}\frac{r(r-1)}{2N(2N-1)}\left(1-\sum_{j=2}^{m_{\text{low}}}(-1)^j\frac{\binom{m_{\text{low}}}{j}}{m_{\text{low}}}\frac{\binom{2N-2j}{r-2j}}{\binom{2N-2}{r-2}\binom{2N}{r}}\right)=\Theta\left(\frac{m}{N^2}r^2\right),\]
the query complexity for the main algorithm is
\[\tilde{O}\left(\mathsf{S}+\frac{1}{\epsilon}\frac{1}{\sqrt{\LBDofFraction}}\left(\frac{1}{\sqrt{\delta}}\mathsf{U}+\mathsf{C}\right)\right)=\tilde{O}\left(r+\frac{1}{\epsilon}\frac{N}{\sqrt{m}}\frac{1}{\sqrt{r}}\right).\]

We finally explain how to choose $r$. Below we assume that $\epsilon>1/m_{\text{up}}$. This assumption can be done without loss of generality since smaller values of $\epsilon$ do not lead to better approximation of $m$ (since $m$ is an integer). For convenience we also assume that $m_{\text{up}}$ is such that $N/m_{\text{up}}=\Omega(\log N)$: for larger values of $m_{\text{up}}$ we can simply use the classical algorithm from Appendix~\ref{sec:classical}. 

Take $r=\Theta\left(\epsilon^{1/12}\left(\frac{N}{\sqrt{m_{\text{up}}}}\right)^{2/3}\right)=\Theta\left(\epsilon^{1/12}\left(\frac{N}{\sqrt{m}}\right)^{2/3}\right)$ such that 
$R_1<\sqrt{\epsilon /2}$, $R^\prime\leq\sqrt{\epsilon}$, $r\ge 1$ and $r\le 2N$.
Such an $r$ exists since 
\begin{align*}
\Theta\left(\frac{m}{N^2}\left(\epsilon^{1/12}\left(\frac{N}{\sqrt{m}}\right)^{2/3}\right)^2\right)&=\Theta\left(\frac{m^{1/3}}{N^{2/3}}\cdot\epsilon^{1/6}\right)\\
&=O\left(\frac{1}{N^{1/3}(\log N)^{1/3}}\cdot\epsilon^{1/6}\right)\\
&=O\left(\frac{\epsilon^{1/3}}{(\log N)^{1/3}}\cdot \epsilon^{1/6}\right)=O\left(\frac{\epsilon^{1/2}}{(\log N)^{1/3}}\right)
\end{align*}
and
\[
\Theta\left(\epsilon^{1/12}\left(\frac{N}{\sqrt{m}}\right)^{2/3}\right)=\Omega\left(\left(\frac{1}{m}\right)^{1/12}\left(\frac{N}{\sqrt{m}}\right)^{2/3}\right)=\Omega\left(\left(\frac{N}{m}\right)^{2/3}\right)=\Omega\left((\log N)^{2/3}\right).
\]
(The third condition $r\le 2N$ is trivially satisfied.)
For this choice of $r$ the complexity of the main algorithm becomes $\tilde{O}\left(\frac{1}{\epsilon^{25/24}}\left(\frac{N}{\sqrt{m}}\right)^{2/3}\right)$.

The overall query complexity (including the complexity of computing the constant-factor approximation) is
\[\tilde{O}\left(\frac{1}{\epsilon^{25/24}}\left(\frac{N}{\sqrt{m}}\right)^{2/3}+\left(\frac{N}{\sqrt{\bar{m}}}\right)^{2/3}\right),\]
as claimed.
\end{proof}

\section*{Acknowledgments}
 FLG was supported by JSPS KAKENHI grants Nos.~JP19H04066, JP20H05966, JP20H00579, JP20H04139, JP21H04879 and MEXT Quantum Leap Flagship Program (MEXT Q-LEAP) grants No.~JPMXS0118067394 and JPMXS0120319794. 
 
\medskip
\printbibliography

\appendix

\section{Classical Algorithm for Approximate Collision Counting}
\label{sec:classical}
\noindent
In this appendix we describe the ``folklore'' classical algorithm based on random sampling for approximating the number of collisions $m$ of two injective functions $f,g\colon\{1,\ldots,N\}\to\{1,\ldots,K\}$, which has been used implicitly in~\cite{Andoni+SODA10,Naumovitz+SODA17}.
More precisely, we show that given a lower bound $\bar{m}\leq m$, there exists a randomized algorithm which outputs a value $\hat{m}$ such that $|m-\hat{m}|<\epsilon m$ holds with high probability, using
\[O\left(\frac{1}{\epsilon}\frac{N}{\sqrt{m}}+\frac{N}{\sqrt{\bar{m}}}\right)\]
evaluations of $f$ and $g$.

Let us write $N_f, N_g=\{1,\ldots,N\}$. The main procedure, which is parametrized by a real number $p\in(0,1)$, is given as Algorithm~\ref{algo:classical} below.

\begin{algorithm}
\caption{Classical Approximate Collision Counting (here $p\in(0,1)$ is a parameter)}
\label{algo:classical}
\SetKwInOut{Input}{Input}
\DontPrintSemicolon
$S_f, S_g\leftarrow\emptyset$.\;
\For{{\bf each} $i_f\in N_f$}{include $i_f$ to $S_f$ with probability $p$.}
\For{{\bf each} $i_g\in N_g$}{include $i_g$ to $S_g$ with probability $p$.}
\For{{\bf each} $i_f\in S_f$}{query $f(i_f)$.}
\For{{\bf each} $i_g\in S_g$}{query $g(i_g)$.}
Compute the number $m_S$ of collisions in $S_f\times S_g$.\;
Output $\hat{m}=\frac{m_S}{p^2}$.
\end{algorithm}

The complexity of Algorithm~\ref{algo:classical} is $O(|S_f|+|S_g|)$ queries. Since the expected value of the size of these sets is $\operatorname{E}[|S_f|+|S_g|]=2Np$, the expected query complexity of Algorithm~\ref{algo:classical} is $O(Np)$.
We now analyze its correctness.

\begin{lemma}
\label{lemma:Chernoff}
If $p\geq\frac{1}{\epsilon}\sqrt{\frac{3}{m}\log{\frac{2}{\nu}}}$ for some $\nu$, then with probability at least $1-\nu$ the output $\hat m$ of Algorithm~\ref{algo:classical} is such that $|m-\hat{m}|<\epsilon m$.
\end{lemma}

\begin{proof}
Since each element is included into $S_f$ or $S_g$ with probability $p$, the probability for a collision pair to be included is $p^2$.
Then the expected value of the number of collisions $j$ in $S_f\times S_g$ is $\operatorname{E}[j]=mp^2$.
With $m^\prime =\frac{j}{p^2}=\frac{j}{\operatorname{E}[j]}m$, the condition $|m-m^\prime|<\epsilon m$ is established if and only if $|\frac{j}{\operatorname{E}[j]}-1|<\epsilon$ holds.
By the use of Chernoff-Hoeffding bound~\cite{Randomized} we have
\begin{align*}
    \operatorname{Pr}\left[\frac{j}{\operatorname{E}[j]}>(1+\epsilon)\right] &\leq e^{-\frac{\epsilon^2}{3}\operatorname{E}[j]}.\\
    \operatorname{Pr}\left[\frac{j}{\operatorname{E}[j]}<(1-\epsilon)\right] &\leq e^{-\frac{\epsilon^2}{2}\operatorname{E}[j]}.
\end{align*}
If $p\geq\frac{1}{\epsilon}\sqrt{\frac{3}{m}\log{\frac{2}{\nu}}}$ for some $\nu$, then $\nu\geq 2e^{-\frac{\epsilon^2}{3}mp^2}=2e^{-\frac{\epsilon^2}{3}\operatorname{E}[j]}$, which implies that the probability of success is at least $1-\nu$ if $p\geq\frac{1}{\epsilon}\sqrt{\frac{3}{m}\log{\frac{2}{\nu}}}$.
\end{proof}

We now describe the whole algorithm. Let $\delta$ be a small constant. First, we apply 
Algorithm~\ref{algo:classical} with  \[
p=\frac{1}{\epsilon_0}\sqrt{\frac{3}{\bar{m}}\log{\frac{2}{\delta/2}}}
\]
for a constant $\epsilon_0>0$.
It outputs $\hat{m}_1$ such that $|m-\hat{m}_1|<\epsilon_0$ with probability at least $1-\delta/2$.
Then we apply Algorithm~\ref{algo:classical} with \[
p=\frac{1}{\epsilon}\sqrt{\frac{3}{\hat m_1}\log{\frac{2}{\delta/2}}}.
\]
It outputs $\hat{m}$ such that $|m-\hat{m}|<\epsilon m$ with probability at least $1-\delta/2$.
The overall expected query complexity is
\[O\left(\frac{1}{\epsilon}\frac{N}{\sqrt{m}}\left(\log\frac{1}{\delta}\right)^{1/2}+\frac{N}{\sqrt{\bar{m}}}\left(\log\frac{1}{\delta}\right)^{1/2}\right),\]
and the probability of success is at least $(1-\delta/2)(1-\delta/2)\ge 1-\delta$. 
This upper bound on the expected query complexity can be converted into an upper bound on the worst-case query complexity using a tail bound.

Note that the dependence on $N$ and $m$ in the above complexity is optimal. Indeed, by an argument similar to the argument used in the proof of Lemma~\ref{lemma:Chernoff}, it is easy to show that $\Omega(N/\sqrt{m})$ queries are necessary to hit a collision.

\section{Quantum Constant-Factor Approximate Collision Counting}
\label{sec:rough}
\noindent
In this appendix we describe a simple algorithm that uses Ambainis' element distinctness algorithm~\cite{Ambainis07} instead of the classical strategy of Steps 6-10 of Algorithm~\ref{algo:classical}, and computes a constant-factor approximation of the number of collisions $m$ using $\tilde O((N/\sqrt{\bar{m}})^{2/3})$ queries. Here is the formal statement of our result (again we assume that a lower bound $\bar m\le m$ is known and state the complexity in term of $\bar{m}$).

\begin{theorem}\label{th:rough}
There exists a quantum algorithm that uses $\tilde O((N/\sqrt{\bar{m}})^{2/3})$ queries and outputs with probability at least $1-1/\textrm{poly}(N)$ an estimator $\hat m$ such that $|m-\hat m|\le m/2$.
\end{theorem}

The quantum algorithm is described as Algorithm~\ref{algo:rough} below, where again we write $N_f, N_g=\{1,\ldots,N\}$.

\begin{algorithm}
\caption{Quantum Constant-Factor Approximate Collision Counting}
\label{algo:rough}
\SetKwInOut{Input}{Input}
\DontPrintSemicolon
$\hat m\leftarrow 0$, $p\leftarrow 2\sqrt{\frac{3}{N}\log{(2N)}}$.\;
\While{$p<4\sqrt{\frac{3}{\bar{m}}\log{(2N)}}$}{
$S_f, S_g\leftarrow\emptyset$.\;
\For{{\bf each} $i_f\in N_f$}{include $i_f$ to $S_f$ with probability $p$.}
\For{{\bf each} $i_g\in N_g$}{include $i_g$ to $S_g$ with probability $p$.}
\If{$|S_f|\le 10\log N\cdot Np$ and $|S_g|\le 10\log N\cdot Np$} 
{Compute the number $m_S$ of collisions in $S_f\times S_g$.\;
\If{$m_S< \lfloor100(\log N)^2\rfloor$}{$\hat{m}\leftarrow\frac{m_S}{p^2}$.}
}
$p\leftarrow 2p$.
}
Output $\hat m$.
\end{algorithm}

At Steps 9-10, observe that we actually only need to check whether $m_S< \lfloor100(\log N)^2\rfloor$ holds. We thus simply apply Ambainis' element distinctness algorithm~\cite{Ambainis07} consecutively $k=\lfloor100(\log N)^2\rfloor$ times, removing the collisions as soon as they are discovered, and check if we obtain strictly less than $k$ collisions, in which case we decide that $m_S< \lfloor100(\log N)^2\rfloor$. The complexity of Steps 9-10 is thus $\tilde O((|S_f|+|S_g|)^{2/3})=\tilde O((Np)^{2/3})=\tilde O((N/\sqrt{\bar{m}})^{2/3})$ queries. Since Steps 9-10 are repeated $O(\log N)$ times, and all other steps can be implemented without any queries, the overall query complexity of Algorithm~\ref{algo:rough} is $\tilde O((N/\sqrt{\bar{m}})^{2/3})$.

We now analyze the correctness of the algorithm. First, observe that since $\operatorname{E}[|S_f|
]=\operatorname{E}[|S_g|]=Np$, the two conditions of Step 8 hold with probability at least $1-1/\textrm{poly}(N)$, from Chernoff bound. Since the success probability of Ambainis' element distinctness algorithm can also be made larger than $1-1/\textrm{poly}(N)$ (e.g., by repeating Ambainis' algorithm $O(\log N)$ times), the implementation of test of Step 10 described in the previous paragraph is correct with probability at least $1-1/\textrm{poly}(N)$. Below, we assume that the two conditions of Step~8 hold and that the test of Step 10 is implemented correctly, and show that Algorithm~\ref{algo:rough} outputs with high probability a good approximation of $m$ under these assumptions.


Lemma~\ref{lemma:Chernoff} guarantees that as soon as the inequality
\begin{equation}\label{ineq1}
p\geq2\sqrt{\frac{3}{m}\log{(2N)}}
\end{equation} 
becomes true,
the number of collisions $m_S$ at Step 10 satisfies $|m-m_S/p^2|< m/2$ with probability at least $1-1/N$. Inequality~(\ref{ineq1}) necessary becomes true during the while loop of Algorithm~\ref{algo:rough} and once it becomes true, it remains true for all the subsequent values of $p$. Let $p_0$ denote the first value taken by $p$ during the while loop such the Inequality~(\ref{ineq1}) holds. The correctness of Algorithm~\ref{algo:rough} then follows from the following lemma.

\begin{lemma}
For the value $p=p_0$, we have 
\[
\Pr[m_S< \lfloor100(\log N)^2]\ge 1-1/\textrm{poly}(N).
\]
\end{lemma}

\begin{proof}
Notice that $p_0\in\left[\:2\sqrt{\frac{3}{m}\log{(2N)}},\:4\sqrt{\frac{3}{m}\log{(2N)}}\:\right)$. The expected value of $m_S$ is thus
\[
\operatorname{E}[m_S]=mp_0^2\le 48\log{(2N)}
\]
and the claim follows from Chernoff bound.
\end{proof}

\end{document}